\documentclass[11pt]{article}
\usepackage[authoryear,round]{natbib}
\usepackage{float}
\usepackage{graphicx,amsmath,amssymb,mathrsfs,amsfonts,amsthm,color,slashbox,amsbsy,ragged2e}
\usepackage{hyperref}
\usepackage{subfigure}
\usepackage{epsfig}
 \usepackage{lineno}

\newcommand{\Gb}{{\bar{G}}}
\newcommand{\Fb}{{\bar{F}}}

\newcommand{\la}{\lambda}

\newcommand{\R}{\mathbb{R}}
\newcommand{\lab}{\boldsymbol{\la}}
\newcommand{\pb}{\boldsymbol{p}}
\newcommand{\ub}{\boldsymbol{u}}

\newcommand{\bq}{\begin{equation}}
\newcommand{\eq}{\end{equation}}
\newcommand{\bqs}{\begin{equation*}}
\newcommand{\eqs}{\end{equation*}}
\newcommand{\bqa}{\begin{eqnarray}}
\newcommand{\eqa}{\end{eqnarray}}
\newcommand{\bqas}{\begin{eqnarray*}}
\newcommand{\eqas}{\end{eqnarray*}}
\newcommand{\bc}{\begin{cases}}
\newcommand{\ec}{\end{cases}}
\newcommand{\bt}{\begin{thm}}
\newcommand{\et}{\end{thm}}

\newtheorem{theorem}{Theorem}[section]
\newtheorem{lemma}{Lemma}[section]
\newtheorem{corollary}{Corollary}[section]

\newtheorem{definition}{Definition}[section]

\title{Stochastic comparisons between the extreme claim amounts from two heterogeneous portfolios in the case of transmuted-G model}
\author{ Hossein Nadeb, Hamzeh Torabi, Ali Dolati\\
Department  of Statistics, Yazd University,  Yazd, Iran,\\
}
\date{}

\begin{document}
\maketitle

\begin{abstract}
Let $ X_{\la_1}, \ldots , X_{\la_n}$ be independent non-negative
random variables belong to the transmuted-G model and let
$Y_i=I_{p_i} X_{\la_i}$, $i=1,\ldots,n$, where $I_{p_1}, \ldots,
I_{p_n}$ are independent Bernoulli random variables independent of
$X_{\la_i}$'s, with ${\rm E}[I_{p_i}]=p_i$, $i=1,\ldots,n$. In
actuarial sciences, $Y_i$ corresponds to the claim amount in a
portfolio of risks. In this paper we compare the smallest and the
largest claim amounts of two sets of
independent portfolios belonging to the transmuted-G model, in the
sense of usual stochastic order, hazard rate order and dispersive
order, when the variables in one set have the parameters
$\lambda_1,...,\lambda_n$ and the variables in the other set  have
the parameters $\lambda^{*}_1,...,\lambda^{*}_n$. For illustration
we apply the results to the transmuted-G exponential and the transmuted-G Weibull
models.
\end{abstract}


\section{Introduction}\label{}
Annual premium is the amount paid by the policyholder as the cost
of the insurance cover being purchased. Indeed, it is the primary
cost to the policyholder for assigning the risk to the insurer
which depends on the type of insurance. Determination of the
annual premium is one of the important problem in insurance
analysis. For this purpose, the smallest and the largest claim
amounts play an important role in providing useful information. An
attractive problem for the actuaries is expressing preferences
between random future gains or losses (Barmalzan et al.
\cite{bar2}). For this purpose, stochastic orderings are very
helpful. Stochastic orderings have been extensively used in some
areas of sciences such as management science, financial economics,
insurance, actuarial science, operation research, reliability
theory, queuing theory and survival analysis. For more details on
stochastic orderings we refer to M\"{u}ller and Stoyan
\cite{must}, Shaked and Shanthikumar \cite{ss} and Li and Li
\cite{lll}. The transmuted-G $ ({\rm TG})$ model, which introduced
by Mirhossaini and Dolati \cite{mir2} and Shaw and Buckley
\cite{shbu}, is an attractive model for constructing new flexible
distributions. Let $F$ be an absolutely continuous distribution
function with the corresponding survival function $\Fb$. The
random variables $X_{\la}$ said to belong to the ${\rm TG}$ model
with the baseline distribution function $F$, if $X_{\la}$ has the
distribution function
\begin{equation*}
F_{X_{\la}}(x)=F(x)\left(1+\la \Fb (x) \right),
\end{equation*}
where $-1\leq \lambda\leq 1$. We use the notion $ X_{\la}
\thicksim {\rm TG}(\la)$ for the transmuted-G model.

Several distributions have been generalized by this transmuting
approach in the literature. Some of them are the transmuted
Weibull distribution by Aryal and Tsokos \cite{arts}, the
transmuted Maxwell distribution by Iriarte and Astorga
\cite{iras}, the transmuted linear exponential distribution by
Tian et al. \cite{tian}, the transmuted log-logistic distribution
by Granzotto and Louzada \cite{gr}, the transmuted Dagum
distribution by Elbatal and Aryal \cite{elb}, the transmuted
Erlang-truncated exponential distribution by Okorie et al.
\cite{ok}, the transmuted exponentiated Weibull geometric
distribution by Saboor et al. \cite{sab}, the transmuted
exponential Pareto distribution by Al-Babtain \cite{alba}, the
transmuted two-parameter Lindley distribution by Kemaloglu and
Yilmaz \cite{keyi} and the transmuted Birnbaum-Saunders
distribution by Bourguignon et al. \cite{bou}.

The problem of stochastic comparisons of some quantities such as
the number of claims, the aggregate claim amounts, the smallest
and the largest claim amounts in two portfolios, have been
considered by many researches in literature; see, e.g., Karlin and
Novikoff \cite{kar}, Ma \cite{ma}, Frostig \cite{fro}, Hu and Ruan
\cite{huru}, Denuit and Frostig \cite{defr}, Khaledi and Ahmadi
\cite{khah}, Zhang and Zhao \cite{zz}, Barmalzan et al.
\cite{bar1}, Li and Li \cite{lili}, Barmalzan and Najafabadi
\cite{bana}, Barmalzan et al. \cite{bar3}, Barmalzan et al.
\cite{bar2} and Balakrishnan et al. \cite{baet}. Flexibility of the transmuted-G model is a good property to assuming this model as the distribution of severities in insurance. Motivated by the
extensive applications of the transmuted-G family to make flexible
models from a given baseline distribution, in this paper we study
stochastic comparisons between the extreme claim amounts from two
heterogeneous portfolios in the case of transmuted-G model. To be
exact, suppose that $X_{\la}$ denotes the total random severities
of a policyholder in an insurance period, and let $I_{p}$ be a
Bernoulli random variable associated with $X_{\la}$, such that
$I_{p}=1$ whenever the policyholder makes random claim amounts
$X_{\la}$ and $I_{p}=0$ whenever does not make a claim. In this
notation, $Y= I_{p} X_{\la}$ is the claim amount in a portfolio of
risks. Consider two sets of heterogeneous portfolios $X_{\la_1},
\ldots , X_{\la_n} $ and $ X_{\la^{*}_1}, \ldots , X_{\la^{*}_n} $
belonging to the TG model and let $Y_i=I_{p_i} X_{\la_i}$ and
$Y^*_i=I_{p^{*}_i} X_{\la^*_i}$, $i=1,\ldots,n$, where $I_{p_i}$
independent of $X_{\lambda_i}$ and $I_{p^{*}_i}$ independent of
$X_{\lambda^{*}_i}$ are independent Bernoulli random variables
with ${\rm E}[I_{p_i}]=p_i$ and ${\rm E}[I_{p^*_i}]=p^*_i$. Let
$Y_{1:n}=\min(Y_1,\ldots,Y_n)$,
$Y^*_{1:n}=\min(Y^*_1,\ldots,Y^*_n)$,
$Y_{n:n}=\max(Y_1,\ldots,Y_n)$ and
$Y^*_{n:n}=\max(Y^*_1,\ldots,Y^*_n)$ be the smallest and the largest claim amounts, arise from $Y_1,\ldots,Y_n$ and
$Y^*_1,\ldots,Y^*_n$. In this paper we compare $Y_{1:n}$ and
$Y^{*}_{1:n}$ in the sense of the usual stochastic order, hazard rate
order and dispersive order and  $Y_{n:n}$ and $Y^{*}_{n:n}$ in the
sense of the usual stochastic order and hazard rate order. For illustration we apply the results to the transmuted-G
exponential and the transmuted-G Weibull models. The rest of the
paper is organized as follows. In Section \ref{sec2}, we recall
some definitions and lemmas which will be used in the sequel. In
Section \ref{sec3}, stochastic comparisons of the largest claim
amounts from two heterogeneous portfolios of risks in a
transmuted-G model in the sense of the usual stochastic ordering
and reversed hazard rate ordering are discussed. In Section
\ref{sec4}, stochastic comparisons of the smallest claim amounts
from two heterogeneous portfolios of risks in a transmuted-G model
in the sense of the usual stochastic ordering, hazard rate
ordering and dispersive ordering are discussed. In Section \ref{sec6} we consider the transmuted-G
exponential and the transmuted-G  Weibull models for illustration
of the established results.

\section{The basic definitions and some prerequisites}\label{sec2}
In this section, we recall some notions of stochastic orderings,
majorization, weakly majorization and related orderings and  some
useful lemmas which are helpful to  prove the main results.
Throughout the paper, we use the notations $ \Bbb R =
(-\infty,+\infty) $, $ \Bbb R_{+} = [0,+\infty) $ and $ \Bbb
R_{++} = (0,+\infty) $. The term increasing (decreasing) is used
for monotone nondecreasing (nonincreasing). Let $X$ and $Y$ be two
non-negative random variables with the respective distribution
functions $ F $ and $ G $, the density functions $ f $ and $ g $,
the survival functions $ \bar{F}=1 - F $ and $ \bar{G}=1 - G  $,
the right continuous inverses $F^{-1} $ and $ G^{-1} $, the hazard
rate functions $ r_{X}={f}/{\bar{F}} $ and $ r_{Y}={g}/{\bar{G}}
$, and the reversed hazard rate functions $
\tilde{r}_{X}={f}/{F} $ and $ \tilde{r}_{Y}={g}/{G} $.
\begin{definition}
{\rm $ X $ is said to be smaller than $ Y $ in the
\begin{itemize}
\item[{\rm(i) }] usual stochastic ordering, denoted by $ X \leq_{\rm st} Y $, if $ \bar{F}(x)\leq\bar{G}(x) $ for all  $ x \in \Bbb R$,

\item[{\rm(ii)}] hazard rate ordering, denoted by $ X\leq_{\rm hr}Y $, if $\Gb(x)/\Fb(x)$ is increasing in $x \in \Bbb R$, or $ r_{Y} (x) \leq r_{X}(x) $ for all $ x \in \Bbb R$,

\item[{\rm(iii)}] reversed hazard rate ordering, denoted by $ X\leq_{\rm rh}Y $, if $G(x)/F(x)$ is increasing in $x \in \Bbb R_{+}$, or $ \tilde{r}_{X} (x) \leq \tilde{r}_{Y}(x) $ for all $ x \in \Bbb R_{+}$,

\item[\rm (iv)] dispersive ordering, denoted by $ X\leq_{\rm disp}Y $, if $  F^{-1}(\beta)-F^{-1}(\alpha)\leq G^{-1}(\beta)-G^{-1}(\alpha)$ for all $ 0\leq\alpha\leq\beta\leq 1 $.
\end{itemize}
}
\end{definition}

We know that the hazard rate and reversed hazard rate orderings imply the usual stochastic ordering.
\begin{lemma}[Shaked and Shanthikumar\cite{ss}, Theorem 3.B.20]\label{3b20}
{\rm Let $X$ and $Y$ be two non-negative random variables. If $X
\leq_{{\rm hr}} Y$ and $X$ or $Y$ is decreasing failure rate
(DFR), then $X \leq_{{\rm disp}} Y$. }
\end{lemma}
For a comprehensive discussion on various stochastic orderings, we
refer to Li and Li \cite{lll} and Shaked and Shanthikumar
\cite{ss}.

We also need the concept of majorization of vectors and matrices
and the Schur-convexity and Schur-concavity of functions. For a
comprehensive discussion of these topics we refer to Marshall et
al. \cite{met}. We use the notation $ x_{(1)}\leq x_{(2)}\leq
...\leq x_{(n)}$ to denote the increasing arrangement of the
components of the vector $ \boldsymbol{x} = (x_{1}, \ldots ,
x_{n})$.
\begin{definition}
{\rm The vector $ \boldsymbol{x} $ is said to be
\begin{itemize}
\item[ (i)] weakly submajorized by the vector $ \boldsymbol{y} $ (denoted by $ \boldsymbol{x}\preceq_{\rm w}\boldsymbol{y} $) if
$\sum_{i=j}^{n}x_{(i)}\leq \sum_{i=j}^{n}y_{(i)}$ for all $j = 1, \ldots , n $,

\item[ (ii)] weakly supermajorized by the vector $ \boldsymbol{y} $ (denoted by $ \boldsymbol{x}\mathop \preceq \limits^{{\mathop{\rm w}} }\boldsymbol{y} $) if $ \sum_{i=1}^{j}x_{(i)}\geq \sum_{i=1}^{j}y_{(i)} $ for all $ j = 1, \ldots , n $,

\item[ (iii)] majorized by the vector $ \boldsymbol{y} $ (denoted by $ \boldsymbol{x}\mathop \preceq \limits^{{\mathop{\rm m}} }\boldsymbol{y} $) if $ \sum_{i=1}^{n}x_{i}= \sum_{i=1}^{n}y_{i}$ and $\sum_{i=1}^{j}x_{(i)}\geq \sum_{i=1}^{j}y_{(i)}$ for all $j = 1, \ldots , n-1 $.
\end{itemize}}
\end{definition}
\begin{definition}
{\rm A real valued function $ \varphi $ defined on a set $ \mathscr{A}\subseteq {\Bbb R}^{n} $ is said to be Schur-convex (Schur-concave) on $ \mathscr{A} $ if

\[
\boldsymbol{x} \mathop \preceq \limits^{{\mathop{\rm m}} }\boldsymbol{y} \quad \text{on}\quad \mathscr{A} \Longrightarrow \varphi(\boldsymbol{x})\leq (\geq)\varphi(\boldsymbol{y}).
\]}
\end{definition}


\begin{lemma}[Marshall et al.\cite{met}, Theorem 3.A.4]\label{3a4}
{\rm Let $ \mathscr{A}\subseteq \Bbb R $ be an open interval and let $l:  \mathscr{A}^n\rightarrow  \Bbb R$ be continuously differentiable. $l$ is Schur-convex (Schur-concave) on $\mathscr{A}^n$ if and only if, $l$ is symmetric on  $\mathscr{A}^n$ and for all $i\neq j$,
\begin{equation*}
(x_i-x_j)\left(\frac{\partial{l(\boldsymbol{x})}}{\partial{x_i}}-\frac{\partial{l(\boldsymbol{x})}}{\partial{x_j}}\right)\geq(\leq)0, \quad \text{for all}\quad \boldsymbol{x}\in \mathscr{A}^n.
\end{equation*}
}
\end{lemma}

\begin{lemma}[Marshall et al.\cite{met}, Theorem 3.A.8]\label{mkl}
{\rm For a function $ l $ on $ \mathscr{A}\subseteq \Bbb R^{n} $, $ \boldsymbol{x}\preceq_{\rm w}\boldsymbol{y} $ implies $ l(\boldsymbol{x}) \leq (\geq) l(\boldsymbol{y}) $ if and only if
it is increasing (decreasing) and Schur-convex (Schur-concave) on $ \mathscr{A} $.}
\end{lemma}


In the following we recall the concepts of $T$-transform matrix
and chain majorization of matrices. We refer to Marshall et al.
\cite{met} for more details.
 \begin{definition}
{\rm A square matrix is called a
\begin{itemize}
\item[(i)] permutation matrix if each row and each column has a single unit, and all other entries are zero,

\item[(ii)] $T$-transform matrix if it is of the form
$T_{\omega}=\omega I_{n}+(1-\omega)\Pi $, where $ 0\leq\omega
\leq1 $, $ I_{n} $ is an $ n\times n $ identity matrix and $ \Pi $
is a permutation matrix that just interchanges two coordinates.
\end{itemize}

Two $T$-transform matrices said to have the same structure if
their permutation matrices are identical; otherwise they said to
have different structures. }
\end{definition}
In the following definition, we recall a multivariate majorization
notion which will be used in the sequel.

\begin{definition}
{\rm Let $ A = \{a_{ij}\} $ and $ B = \{b_{ij}\} $ be two $ m\times n $ matrices. Then $ A $ is said to be  chain majorized by $ B $, denoted by $ A\ll B $, if there exists a finite set of $ n\times n $ $ T $-transform matrices $ T_{\omega_{1}}, \ldots , T_{\omega_{k}} $ such that $ A = BT_{\omega_{1}}\times \ldots \times T_{\omega_{k}} $.
}
\end{definition}
For $i=1,\ldots,m$ , let $\boldsymbol{a}^{R}_i$ and
$\boldsymbol{b}^{R}_i$, denote the $i$th row of $A$ and $B$,
respectively. Then we have
\begin{table}[H]
\begin{center}
\begin{tabular}{ccccc}
$A\ll B$&$\Rightarrow$&$\boldsymbol{a}^{R}_i\mathop \preceq \limits^{{\mathop{\rm m}} }\boldsymbol{b}^{R}_i$&$\Rightarrow$&$\boldsymbol{a}^{R}_i \preceq_{\rm w} \boldsymbol{b}^{R}_i$\\
&&$\Downarrow$ &&\\
&&$\boldsymbol{a}^{R}_i\mathop \preceq \limits^{{\mathop{\rm w}} }\boldsymbol{b}^{R}_i$&$\Rightarrow$&$\prod \limits_{j=1}^n b_{ij}\leq \prod \limits_{j=1}^n a_{ij} $,
\end{tabular}
\end{center}
\end{table}
where the last consequence holds whenever $\boldsymbol{a}^{R}_i,
\boldsymbol{b}^{R}_i \in \Bbb R^n_{++}$. Let
\[
S_{n} = \bigg\lbrace (\boldsymbol{x},\boldsymbol{y}) = \begin{bmatrix}
x_{1}, \ldots , x_{n} \\
y_{1}, \ldots , y_{n}
\end{bmatrix} :-1\leq x_{i}\leq 1, y_{i}>0, \text{and}\quad ( x_{i}-x_{j})(y_{i}-y_{j})\leq0,\] $\hspace{2cm} i,j = 1,\dots,n\bigg\}.$

We recall the following lemmas similar to the lemmas in Balakrishnan et al.\cite{bee}, which their proofs are very similar to the proofs of lemmas in Balakrishnan et al.\cite{bee}. So, the proofs are omitted for simplicity.

\begin{lemma}
\label{t2}
{\rm A differentiable function $ \varphi : \Bbb R^{4}\longrightarrow \Bbb R_{+} $ satisfies
\begin{equation}\label{222a}
\varphi(A)\leq\varphi(B) \text{ for all } A,B \text{ such that } B\in S_{2}, \text{ and } A\ll B
\end{equation}
if and only if
\begin{itemize}
\item[(i)] $ \varphi(B)=\varphi(B\Pi) $ for all permutation
matrices $ \Pi $, and all $ B\in S_{2} $; and \item[(ii)] $
\sum_{i=1}^{2} (b_{ik}-b_{ij})(\varphi_{ik}(B)-\varphi_{ij}(B))
\geq0$ for all $ j,k = 1,2 $, and all $ B\in S_{2} $, where $
\varphi_{ij}(B) = \dfrac{\partial\varphi(B)}{\partial b_{ij}} $.
\end{itemize}
}
\end{lemma}
\begin{lemma}
\label{t3} {\rm Let $ \Psi : \Bbb R^{2}\longrightarrow \Bbb
R_{+} $ be a differentiable function, and let the function $
\upsilon_{n} : \Bbb R^{2n}\longrightarrow \Bbb R_{+} $ defined
by
\begin{equation*}
\upsilon_{n}(B) = \prod_{i=1}^{n}\Psi(b_{1i},b_{2i}).
\end{equation*}
If $ \upsilon_{2} $ satisfies \eqref{222a}, then, for $ B\in S_{n}
$, and $ A = BT_{\omega} $, we have $
\upsilon_{n}(A)\leq\upsilon_{n}(B) $. }
\end{lemma}

\begin{lemma}\label{lem1}
{\rm Let $ k$ be a function defined by
\begin{equation*}
k(x,y,z):=\frac{x}{1-xyz}.
\end{equation*}
Then,
\begin{itemize}
\item[(i)] $ k$ is increasing in $x$,
\item[(ii)] $ k$ is increasing in $y$, when $z\geq 0$.
\end{itemize}
}
\end{lemma}

\section{Results for the largest claim amounts}\label{sec3}
It is clear that the random variables $Y_i=I_{p_i}X_{\la_i}$,
$i=1,\ldots,n$, are discrete-continuous, which are equal to zero
with the probability $1-p_i$, and $X_{\la_i}$ with the probability
$p_i$, $i=1,\ldots, n$. The distribution function and the reversed
hazard rate function of $Y_{n:n}$, the largest claim amount, are
given by
\begin{equation}\label{largest}
G_{Y_{n:n}}(x)=\prod \limits_{i=1}^n \Bigg(1-p_i \Fb (x) \Big(1-\la_i F(x)\Big)\Bigg),\quad x\geq 0.
\end{equation}
and
\begin{equation}\label{rhlargest}
\tilde{r}_{Y_{n:n}}(x)= \sum \limits_{i=1}^{n} \frac{f(x) \Big(1+\la_i (1-2 F(x))\Big) p_i}{1-p_i\Fb(x)\left(1-\la_i F(x)\right)}{\rm I}_{[x>0]}+{\rm I}_{[x=0]},
\end{equation}
respectively; where ${\rm I}_A$ denotes the indicator function.
Similarly, the distribution function and the reversed hazard rate function of $Y^*_{n:n}$ is the same as in \eqref{largest} and \eqref{rhlargest}
upon replacing $\la_i$ by $\la^*_i$ and $p_i$ by $p^*_i$, $i=1,\ldots,n$, respectively.

The following theorem provides a comparison between the largest
claim amounts in two heterogeneous portfolio of risks, in the sense of
the usual stochastic ordering via matrix majorization.
\begin{theorem}\label{th2}
{\rm Let $ X_{\la_1} , X_{\la_2} $ ($ X_{\la^{*}_1} ,
X_{\la^{*}_2} $) be independent non-negative random variables with
$ X_{\la_i} \thicksim {\rm TG}(\la_i)$ ($ X_{\la^{*}_i} \thicksim
{\rm TG}(\la^*_{i} )$), $i = 1, 2 $. Further, suppose that $I_{p_1},
I_{p_2}$  ($I_{p^*_1},I_{p^*_2} $) is a set of independent
Bernoulli random variables, independent of the $X_{\la_i}$'s
($X_{\la^*_i}$'s), with ${\rm E}[I_{p_i}]=p_i$ (${\rm
E}[I_{p^*_i}]=p^*_i$), $i=1,2$. Let $h:[0,1]\rightarrow I\subset
\R_{+}$ be a differentiable and strictly increasing concave
function on $[0,1]$ with the non-zero derivative. Then for $(\lab,h(\pb))\in
S_{2} $, we have
\begin{eqnarray*}
\begin{bmatrix}
\la^*_{1} & \la^*_{2} \\
h(p^*_{1}) & h(p^*_{2})\\
\end{bmatrix}\ll \begin{bmatrix}
\la_{1} & \la_{2} \\
h(p_{1}) & h(p_{2})\\
\end{bmatrix}\Longrightarrow Y^*_{2:2}\leq_{\rm st}Y_{2:2}.
\end{eqnarray*}
}
\end{theorem}
\begin{proof}
In view of \eqref{largest}, the distribution function of $Y_{2:2}$
can be rewritten as
\begin{equation*}
G_{Y_{2:2}}(x)=\prod \limits_{i=1}^2 \Bigg(1-h^{-1}(u_i) \Fb (x) \Big(1-\la_i F(x)\Big)\Bigg),\quad x\geq 0,
\end{equation*}
where $h^{-1}$ is the inverse of the function $h$, and
$u_i=h(p_i)$, $i=1,2$. For fixed $ x\geq0 $, we have to show that
the function $ G_{Y_{2:2}}(x) $ satisfies the conditions of Lemma
\ref{t2}. Clearly,  the condition (i) is satisfied. To check the
condition (ii), consider the function $ \rho $ given by
\begin{equation}\label{rho}
\rho(\lab, \boldsymbol{u}) = \rho_{1}(\boldsymbol{\la}, \boldsymbol{u}) + \rho_{2}(\boldsymbol{\la}, \boldsymbol{u}),
\end{equation}
where
\begin{equation*}
\rho_{1}(\boldsymbol{\la}, \boldsymbol{u}) =(u_{1}-u_{2})\bigg(\dfrac{\partial G_{Y_{2:2}}(x)}{\partial u_{1}}-\dfrac{\partial G_{Y_{2:2}}(x)}{\partial u_{2}} \bigg),
\end{equation*}
and
\begin{equation*}
\rho_{2}(\boldsymbol{\la}, \boldsymbol{u}) = (\la_{1}-\la_{2})\bigg(\dfrac{\partial G_{Y_{2:2}}(x)}{\partial \la_{1}}-\dfrac{\partial G_{Y_{2:2}}(x)}{\partial \la_{2}} \bigg).
\end{equation*}
The partial derivatives of $ G_{Y_{2:2}}(x) $ with respect to $
u_{i} $ and $ \la_{i} $ are given by
\begin{equation*}
\dfrac{\partial G_{Y_{2:2}}(x)}{\partial u_{i}} =- \frac{(1-\la_i F(x))\frac{\partial h^{-1}(u_i)}{\partial u_i}}{1-h^{-1}(u_i)\Fb(x)\left(1-\la_i F(x)\right)} \Fb(x) G_{Y_{2:2}}(x),
\end{equation*}
and
\begin{equation*}
\dfrac{\partial G_{Y_{2:2}}(x)}{\partial \la_{i}}
=\frac{h^{-1}(u_i)}{1-h^{-1}(u_i)\Fb(x)\left(1-\la_i F(x)\right)}
F(x) \Fb(x) G_{Y_{2:2}}(x).
\end{equation*}
Thus
\begin{eqnarray*}
\rho_{1}(\boldsymbol{\la}, \boldsymbol{u}) = -(u_1-u_2)  \Fb(x) G_{Y_{2:2}}(x)\left (\eta_1( \la_1,u_1)\frac{\partial h^{-1}(u_1)}{\partial u_1}-\eta_1( \la_2,u_2)\frac{\partial h^{-1}(u_2)}{\partial u_2}\right),
\end{eqnarray*}
where, $\eta_1(\la,u)=k\left(1-\la F(x), h^{-1}(u),\Fb(x)\right)$,
and $k$ is the function defined in Lemma \ref{lem1}. The
assumption $ (\boldsymbol{\la},\ub) \in S_{2} $ implies that $
(\la_{1}-\la_{2})(u_{1}-u_{2})\leq0 $ or equivalently, $
\la_{1}\leq\la_{2} $ and $ u_{1}\geq u_{2} $, or $
\la_{1}\geq\la_{2} $ and $u_{1}\leq u_{2} $. We only state the
proof for the case $ \la_{1}\leq\la_{2} $ and $ u_{1}\geq u_{2} $.
The other case is analogously proven. Since $h$ is strictly
increasing and concave then $h^{-1}$ is strictly increasing and
convex. The convexity of $h^{-1}$ implies that
\begin{equation*}
0\leq \frac{\partial h^{-1}(u_2)}{\partial u_2}\leq\frac{\partial
h^{-1}(u_1)}{\partial u_1}.
\end{equation*}
In view of Lemma \ref{lem1} the function $\eta_1$ is decreasing in
$\la$ and increasing in $u$, so that
\begin{equation*}
0\leq \eta_1(\la_2,u_2)\leq \eta_1(\la_1,u_2)\leq
\eta_1(\la_1,u_1),
\end{equation*}
which implies that
\begin{equation}\label{rho1}
\rho_1(\lab,\ub)\leq 0.
\end{equation}
On the other hand,
\begin{eqnarray*}
\rho_{2}(\boldsymbol{\la}, \boldsymbol{u}) = (\la_1-\la_2)  F(x) \Fb(x) G_{Y_{2:2}}(x)\Big (\eta_2(\la_1, u_1)-\eta_2(\la_2, u_2)\Big),
\end{eqnarray*}
where, $\eta_2(\la,u)=k\left( h^{-1}(u),1-\la F(x),\Fb(x)\right)$.
By a similar argument the function $\eta_2$ is decreasing in $\la$
and increasing in $u$ and
\begin{equation*} \eta_2(\la_2,u_2)\leq
\eta_2(\la_1,u_2)\leq \eta_2(\la_1,u_1),
\end{equation*}
which implies that
\begin{equation}\label{rho2}
\rho_{2}(\boldsymbol{\la}, \boldsymbol{u})\leq 0.
\end{equation}
By using the inequalities \eqref{rho}, \eqref{rho1} and
\eqref{rho2}, we have that
\begin{equation*}
\rho(\boldsymbol{\la}, \boldsymbol{u})\leq 0,
\end{equation*}
and the function $ G_{Y_{2:2}}(x) $ satisfies the condition (ii)
of Lemma \ref{t2}. Now Lemma \ref{t2} and the condition
$(\lab^*,\ub^*)\ll (\lab,\ub)$ implies that
\begin{equation*}
G_{Y_{2:2}}(x) \leq G_{Y^*_{2:2}}(x),
\end{equation*}
which is the required result.
\end{proof}



The following result provides a lower bound for the survival
function of the largest claim amount based on a heterogeneous
portfolio of risks in terms of the survival function of largest
claim amounts based on a homogeneous portfolio of risks.

\begin{corollary}
{\rm Let  $\bar{\la}=\frac{1}{2}(\la_1+\la_2)$ and
$\overline{h(p)}=\frac{1}{2}(h(p_{1})+h(p_{2}))$. Under the
conditions of Theorem \ref{th2} we have
}
\begin{equation*}
\Gb_{Y_{2:2}}(x)\geq 1-\Bigg(1-h^{-1}(\overline{h(p)}) \Fb (x)
\Big(1-\bar{\la} F(x)\Big)\Bigg)^2.
\end{equation*}

\end{corollary}
\begin{proof}
It is clear that $ \begin{bmatrix}
\bar{\la}& \bar{\la} \\
\overline{h(p)}&\overline{h(p)}\\
\end{bmatrix}
=\begin{bmatrix}
\la_{1}& \la_{2} \\
h(p_{1}) & h(p_{2})\\
\end{bmatrix}T_{0.5}$. Thus we have
$ \begin{bmatrix}
\bar{\la}& \bar{\la} \\
\overline{h(p)}&\overline{h(p)}\\
\end{bmatrix}
\ll \begin{bmatrix}
\la_{1}& \la_{2} \\
h(p_{1}) & h(p_{2})\\
\end{bmatrix}$. Now Theorem \ref{t2} gives the required result.
\end{proof}

The following result generalizes the result of Theorem \ref{t2}
for an arbitrary number of random variables.

\begin{theorem}\label{th3}
{\rm Let $ X_{\la_1}, \ldots , X_{\la_n} $ ($ X_{\la^{*}_1},
\ldots , X_{\la^{*}_n} $) be independent non-negative random
variables with $ X_{\la_i} \thicksim {\rm TG}(\la_i)$ ($
X_{\la^{*}_i} \thicksim {\rm TG}(\la^*_{i} )$), $i = 1, \ldots, n
$. Further, suppose that $I_{p_1}, \ldots, I_{p_n}$  ($I_{p^*_1},
\ldots, I_{p^*_n} $) is a set of independent Bernoulli random
variables, independent of the $X_{\la_i}$'s ($X_{\la^*_i}$'s),
with ${\rm E}[I_{p_i}]=p_i$ (${\rm E}[I_{p^*_i}]=p^*_i$),
$i=1,\ldots,n$. Let $h:[0,1]\rightarrow I\subset \R_{+}$ be a
differentiable and strictly increasing concave function on
$[0,1]$, with non-zero derivative. Then for $(\lab,h(\pb))\in
S_{n} $, we have
\begin{eqnarray*}
\begin{bmatrix}
\la^*_{1}&\ldots & \la^*_{n} \\
h(p^*_{1})&\ldots & h(p^*_{n})\\
\end{bmatrix}= \begin{bmatrix}
\la_{1}&\ldots & \la_{n} \\
h(p_{1})&\ldots & h(p_{n})\\
\end{bmatrix}T_{\omega} \Longrightarrow Y^*_{n:n}\leq_{\rm st}Y_{n:n}.
\end{eqnarray*}
}
\end{theorem}
\begin{proof}
{\rm Using Lemma \ref{t3} and Theorem \ref{th2}, we immediately obtain the required result.}
\end{proof}

According to Balakrishnan et al.\cite{bee}, a finite product of
$T$-transform matrices with the same structure is also a
$T$-transform matrix. Thus the following result is a direct
consequence of Theorem \ref{th3}.

\begin{corollary}\label{th4}
{\rm Under the assumptions of Theorem \ref{th3}, for $(\lab,h(\pb))\in S_{n} $, we have

\begin{eqnarray*}
\begin{bmatrix}
\la^*_{1}&\ldots & \la^*_{n} \\
h(p^*_{1})&\ldots & h(p^*_{n})\\
\end{bmatrix}= \begin{bmatrix}
\la_{1}&\ldots & \la_{n} \\
h(p_{1})&\ldots & h(p_{n})\\
\end{bmatrix}T_{\omega_1} \ldots T_{\omega_m}  \Longrightarrow Y^*_{n:n}\leq_{\rm st}Y_{n:n},
\end{eqnarray*}
where $T_{\omega_i}$, $i=1,\ldots,m$, have the same structure.
}
\end{corollary}

The following corollary provides a result for the case where the
 $T$-transform matrices have different structures.

\begin{corollary}\label{cor1}
{\rm Under the assumptions of Theorem \ref{th3}, for  $(\lab,h(\pb))\in S_{n} $,
and $(\lab,h(\pb))T_{\omega_1},\ldots,T_{\omega_i}\in S_{n}$, for
$i=1,\ldots,m-1$, where $m\geq 2$, we have
\begin{eqnarray*}
\begin{bmatrix}
\la^*_{1}&\ldots & \la^*_{n} \\
h(p^*_{1})&\ldots & h(p^*_{n})\\
\end{bmatrix}= \begin{bmatrix}
\la_{1}&\ldots & \la_{n} \\
h(p_{1})&\ldots & h(p_{n})\\
\end{bmatrix}T_{\omega_1} \ldots T_{\omega_m}  \Longrightarrow Y^*_{n:n}\leq_{\rm st}Y_{n:n},
\end{eqnarray*}
}
\end{corollary}

\begin{proof}
{\rm Using Theorem \ref{th3} consecutively, the desired result is immediately obtained.
}
\end{proof}

The following result deals with the comparison of the largest claim amounts in a homogeneous portfolio of risks, in the sense of the reversed hazard rate ordering  via weakly majorization.

\begin{theorem}
{\rm Under the assumptions of
Theorem \ref{th3} with $\la_i=\la^*_i=\la$, for $i=1,\ldots,n$, we have
\begin{eqnarray*}
(h(p^*_1), \ldots, h(p^*_n))  \preceq_{{\rm w}} (h(p_1), \ldots, h(p_n)) \Longrightarrow Y^*_{n:n} \leq_{{\rm rh}}Y_{n:n}.
\end{eqnarray*}
}
\end{theorem}

\begin{proof}
{\rm According to \eqref{rhlargest}, the reversed hazard rate function of $ Y_{n:n} $ can be rewritten as
\begin{equation*}
\tilde{r}_{Y_{n:n}}(x)= \sum_{i=1}^{n} \frac{f(x) \Big(1+\la (1-2 F(x))\Big) h^{-1}(u_i)}{1-h^{-1}(u_i)\Fb(x)\left(1-\la F(x)\right)}{\rm I}_{[x>0]}+{\rm I}_{[x=0]},
\end{equation*}
where, $u_i=h(p_i)$, $i=1,\ldots,n$. First, consider $x=0$. In
this case, $\tilde{r}_{Y_{n:n}}(0)=\tilde{r}_{Y^*_{n:n}}(0)=1$,
and the desired result is obvious. Now, consider $x>0$. Using
Lemma \ref{mkl},  it is enough to show that the function $
\tilde{r}_{Y_{n:n}}(x)$ is Schur-convex and increasing in $ u_{i}
$'s. The partial derivatives of $ \tilde{r}_{Y_{n:n}}(x) $ with
respect to $ u_{i} $ is given by
\begin{eqnarray*}
\frac{\partial \tilde{r}_{Y_{n:n}}(x)}{\partial {u_i}}=\frac{f(x) \Big(1+\la (1-2 F(x))\Big)\frac{\partial h^{-1}(u_i)}{\partial u_i}}{\bigg(1-h^{-1}(u_i)\Fb(x)\left(1-\la F(x)\right)\bigg)^2}\geq0.
\end{eqnarray*}
Thus, $  \tilde{r}_{Y_{n:n}}(x) $ is increasing in each $ u_{i} $.
To prove the Schur-convexity of $  \tilde{r}_{Y_{n:n}}(x) $, from
Lemma \ref{3a4}, it is enough to show that for $ i\neq j $,
\begin{equation*}
(u_{i}-u_{j})\bigg(\frac{\partial \tilde{r}_{Y_{n:n}}(x)}{\partial u_{i}}-\frac{\partial \tilde{r}_{Y_{n:n}}(x)}{\partial u_{j}} \bigg)\geq 0,
\end{equation*}
that is, for $ i\neq j $,
\begin{eqnarray*}
&&(u_i -u_j)f(x) \Big(1+\la (1-2 F(x))\Big)\\
&&\times \Bigg( \frac{\frac{\partial h^{-1}(u_i)}{\partial u_i}}{\bigg(1-h^{-1}(u_i)\Fb(x)\left(1-\la F(x)\right)\bigg)^2}-\frac{\frac{\partial h^{-1}(u_j)}{\partial u_j}}{\bigg(1-h^{-1}(u_j)\Fb(x)\left(1-\la F(x)\right)\bigg)^2}\Bigg)\geq 0.
\end{eqnarray*}
Since $h$ is increasing and concave, then $h^{-1}$ is increasing
and convex. Thus, the inequality immediately holds.}
\end{proof}

\section{Results for the smallest claim amounts}\label{sec4}
It can be easily seen that the survival function and the hazard rate function of $Y_{1:n}$, the smallest claim amount, are given by
\begin{equation}\label{smallest}
\Gb_{Y_{1:n}}(x)= \Bigg(\prod \limits_{i=1}^n p_i \Bigg) \prod \limits_{i=1}^n \Bigg(\Fb (x) \big(1-\la_i F(x)\big)\Bigg),\quad x\geq 0,
\end{equation}
and
\begin{equation}\label{hrsmallest}
r_{Y_{1:n}}(x)= \sum_{i=1}^{n} \frac{ 1+\la_i (1-2 F(x))}{1-\la_i F(x)} r(x) {\rm I}_{[x>0]}+\left(1- \prod \limits_{i=1}^{n} p_i\right){\rm I}_{[x=0]},
\end{equation}
respectively. Similarly, the survival function and the hazard rate function of $Y^*_{1:n}$ is the same as in \eqref{smallest} and \eqref{hrsmallest} upon replacing $\la_i$ by $\la^*_i$ and $p_i$ by $p^*_i$, $i=1,\ldots,n$, respectively.

The following result deals with the comparison of the smallest claim amounts in a portfolio of risks, in the sense of the usual stochastic ordering  via majorization.

\begin{theorem}\label{th5}
{\rm Let $ X_{\la_1}, \ldots , X_{\la_n} $ ($ X_{\la^{*}_1},
\ldots , X_{\la^{*}_n} $) be independent non-negative random
variables with $ X_{\la_i} \thicksim {\rm TG}(\la_i)$ ($
X_{\la^{*}_i} \thicksim {\rm TG}(\la^*_{i} )$), $i = 1, \ldots, n
$. Further, suppose that $I_{p_1}, \ldots, I_{p_n}$  ($I_{p^*_1},
\ldots, I_{p^*_n} $) is a set of independent Bernoulli random
variables, independent of the $X_{\la_i}$'s ($X_{\la^*_i}$'s),
with ${\rm E}[I_{p_i}]=p_i$ (${\rm E}[I_{p^*_i}]=p^*_i$),
$i=1,\ldots,n$. Then, we have
\begin{eqnarray*}
\prod \limits_{i=1}^n p^*_i \leq \prod \limits_{i=1}^n p_i,~ (\la_1, \ldots, \la_n) \preceq_{{\rm w}} (\la^*_1, \ldots, \la^*_n) \Longrightarrow Y^*_{1:n}\leq_{\rm st}Y_{1:n}.
\end{eqnarray*}
}
\end{theorem}

\begin{proof}
{\rm Assume that $\prod \limits_{i=1}^n p^*_i \leq \prod
\limits_{i=1}^n p_i$. Now using \eqref{smallest}, the required
result holds if $X^*_{1:n}\leq_{\rm st}X_{1:n}$, where $X^*_{1:n}$
and $X_{1:n}$ are the smallest order statistics of $(X_{\la^*_1},
\ldots , X_{\la^*_n})$ and $(X_{\la_1}, \ldots , X_{\la_n})$,
respectively. The survival function of $X_{1:n}$ is given by
\begin{equation*}
\Fb_{X_{1:n}}(x)=\prod \limits_{i=1}^n \Bigg(\Fb (x) \big(1-\la_i F(x)\big)\Bigg),\quad x\geq 0.
\end{equation*}
Thus by Lemma \ref{mkl}, it is enough to show that the function $
\Fb_{X_{1:n}}(x)$ is Schur-concave and decreasing in $ \la_{i}
$'s. The partial derivative of $ \Fb_{X_{1:n}}(x) $ with respect
to $ \la_{i} $ is given by
\begin{equation*}
\frac{\partial \Fb_{X_{1:n}}(x)}{\partial \la_i}=-\frac{F(x)\Fb_{X_{1:n}}(x)}{1-\la_i F(x)}\leq 0.
\end{equation*}
Thus $  \Fb_{X_{1:n}}(x)$ is decreasing in each $ \la_{i} $. To
prove the Schur-concavity of $  \Fb_{X_{1:n}}(x)$, from Lemma
\ref{3a4}, it is enough to show that for $ i\neq j $,
\begin{equation*}
(\la_{i}-\la_{j})\bigg(\frac{\partial \Fb_{X_{1:n}}(x)}{\partial \la_{i}}-\frac{\partial \Fb_{X_{1:n}}(x)}{\partial \la_{j}} \bigg)\leq 0,
\end{equation*}
that is, for $ i\neq j $,
\begin{equation*}
-(\la_{i}-\la_{j}) F(x)\Fb_{X_{1:n}}(x)\bigg(\frac{1}{1-\la_i F(x)}-\frac{1}{1-\la_j F(x)} \bigg)\leq 0,
\end{equation*}
which is immediately concluded. }
\end{proof}

The following result provides a lower bound for the survival
function of the smallest claim amount based on a heterogeneous
portfolio of risks in terms of the survival function of smallest
claim amounts based on a homogeneous portfolio of risks.

\begin{corollary}
{\rm Under the assumption of Theorem \ref{th5}
\begin{equation*}
\Gb_{Y_{1:n}}(x)\geq \Bigg(\tilde{p} \Fb (x) \big(1-\tilde{\la}
F(x)\big)\Bigg)^n,
\end{equation*}
where $\prod \limits_{i=1}^n p_i=\tilde{p}^n$ and
$\tilde{\la}=\max  \bigg(\frac{1}{2}(1+\la_1),
\ldots,\frac{1}{2}(1+ \la_n)\bigg)$.}
\end{corollary}
\begin{proof}
It is clear that
\begin{equation*}
 (\la_1, \ldots, \la_n) \preceq_{{\rm w}} \bigg(\frac{1}{2}(1+\la_1), \ldots,\frac{1}{2}(1+ \la_n)\bigg) \preceq_{{\rm w}} (\tilde{\la},\ldots,
 \tilde{\la}).
\end{equation*}
These assumptions satisfy the conditions of Theorem \ref{th5},
which implies the result.
\end{proof}

The following result shows that under the same conditions of Theorem \ref{th5}, a stronger result also holds.

\begin{theorem}\label{th6}
{\rm Under the assumptions of Theorem \ref{th5}, we have
\begin{eqnarray*}
\prod \limits_{i=1}^n p^*_i \leq \prod \limits_{i=1}^n p_i,~ (\la_1, \ldots, \la_n) \preceq_{{\rm w}} (\la^*_1, \ldots, \la^*_n) \Longrightarrow Y^*_{1:n}\leq_{\rm hr}Y_{1:n}.
\end{eqnarray*}
}
\end{theorem}

\begin{proof}
{\rm According to \eqref{smallest}, we have
\begin{equation*}
\frac{\Gb_{Y_{1:n}}(x)}{\Gb_{Y^*_{1:n}}(x)}=\left(\frac{\prod \limits_{i=1}^n p_i}{\prod \limits_{i=1}^n p^*_i}\right)\frac{\Fb_{X_{1:n}}(x)}{\Fb_{X^*_{1:n}}(x)},\quad x\geq 0.
\end{equation*}
}
We have to show that $\frac{\Gb_{Y_{1:n}}(x)}{\Gb_{Y^*_{1:n}}(x)}$ is increasing in $x$, which holds if
\begin{equation}\label{ineq}
1=\lim \limits_{{\mathop{x\rightarrow 0^-}} } \frac{\Gb_{Y_{1:n}}(x)}{\Gb_{Y^*_{1:n}}(x)} \leq  \frac{\Gb_{Y_{1:n}}(0)}{\Gb_{Y^*_{1:n}}(0)}=   \frac{\prod \limits_{i=1}^n p_i}{\prod \limits_{i=1}^n p^*_i},
\end{equation}
and $\frac{\Fb_{X_{1:n}}(x)}{\Fb_{X^*_{1:n}}(x)}$ is increasing in
$x$. Since Inequality \eqref{ineq} holds according to the
assumptions, it is enough to show that  $X^*_{1:n}\leq_{\rm
hr}X_{1:n}$ or equivalently $r_{X_{1:n}}(x)\leq r_{X^*_{1:n}}(x)$,
for $x\geq 0$. The hazard rate function of $X_{1:n}$ is given by
\begin{equation*}
r_{X_{1:n}}(x)=\sum \limits_{i=1}^n \frac{1+\la_i \left(1-2 F(x)\right)}{1-\la_i F(x)} r(x),\quad x\geq0.
\end{equation*}
Thus by Lemma \ref{mkl}, it is enough to show that the function $
r_{X_{1:n}}(x)$ is Schur-convex and increasing in $ \la_{i} $'s.
The partial derivative of $ r_{X_{1:n}}(x) $ with respect to $
\la_{i} $ is given by
\begin{equation*}
\frac{\partial r_{X_{1:n}}(x)}{\partial \la_i}=\frac{\Fb(x) r(x)}{(1-\la_i F(x))^2}\geq 0.
\end{equation*}
Thus $  r_{X_{1:n}}(x)$ is increasing in each $ \la_{i} $. To
prove the Schur-convexity of $  r_{X_{1:n}}(x)$, from Lemma
\ref{3a4}, it is enough to show that for $ i\neq j $,
\begin{equation*}
(\la_{i}-\la_{j})\bigg(\frac{\partial r_{X_{1:n}}(x)}{\partial \la_{i}}-\frac{\partial r_{X_{1:n}}(x)}{\partial \la_{j}} \bigg)\geq 0,
\end{equation*}
that is, for $ i\neq j $,
\begin{equation*}
(\la_{i}-\la_{j}) \Fb(x)r(x)\bigg(\frac{1}{(1-\la_i F(x))^2}-\frac{1}{(1-\la_j F(x))^2} \bigg)\geq 0,
\end{equation*}
where, the inequality is immediately concluded.
\end{proof}

The following result deals with the comparison of the smallest claim amounts in two portfolios of risks, in the sense of the dispersive ordering  via majorization.

\begin{theorem}\label{th7}
{\rm Under the assumptions of Theorem \ref{th5}, if  $F$ is ${\rm DFR}$, $0\leq \la^*_i \leq 1$, $i=1,\ldots,n$, and $f(0)\leq \frac{1-\prod \limits_{i=1}^n p^*_i}{\sum \limits_{i=1}^n (1+\la^*_i)}$, then we have
\begin{eqnarray*}
\prod \limits_{i=1}^n p^*_i \leq \prod \limits_{i=1}^n p_i,~ (\la_1, \ldots, \la_n) \preceq_{{\rm w}} (\la^*_1, \ldots, \la^*_n) \Longrightarrow Y^*_{1:n}\leq_{\rm disp}Y_{1:n}.
\end{eqnarray*}
}
\end{theorem}

\begin{proof}
{\rm By Theorem \ref{th6}, we have that $Y^*_{1:n}\leq_{\rm hr}Y_{1:n}$. According to Mirhossaini et al. \cite{mir1}, $F$ is ${\rm DFR}$ and $0\leq \la^*_i \leq 1$, $i=1,\ldots,n$, imply that the first terms of \eqref{hrsmallest} is decreasing in $x>0$. Therefore,
\begin{equation*}
1-\prod \limits_{i=1}^n p^*_i=r_{Y^*_{1:n}}(0)\geq \lim \limits_{{\mathop{x\rightarrow 0^+}} }r_{Y^*_{1:n}}(x)=\sum \limits_{i=1}^n (1+\la^*_i) r(0)=\sum \limits_{i=1}^n (1+\la^*_i) f(0),
\end{equation*}
} implies that $r_{Y^*_{1:n}}(x)$ is decreasing in $x\geq 0$ and
$Y^*_{1:n}$ is ${\rm DFR}$. Thus, Lemma \ref{3b20} completes the
proof.
\end{proof}
Note that under the assumptions of Theorem \ref{th7}, we can conclude that the variance of $Y^*_{1:n}$ is equal or less than the variance of $Y_{1:n} $.







\section{Application}\label{sec6}
In this section, we provide some special cases for illustration of some results of the paper for $n=3$.
\subsection{Transmuted-G exponential distribution}
Suppose that the baseline distribution in transmuted-G model is
exponential distribution with mean $\theta$. Here this
distribution is denoted by ${\rm TE} (\la,\theta)$. For more
details on this distribution, we refer to Mirhossaini and Dolati
\cite{mir2}.

\begin{itemize}
\item  Let $ X_{\la_1}, X_{\la_2} , X_{\la_3} $ ($ X_{\la^{*}_1}, X_{\la^{*}_2} , X_{\la^{*}_3} $)
 be independent random variables with $ X_{\la_i} \thicksim {\rm TE}(\la_i,0.5)$ ($ X_{\la^{*}_i} \thicksim {\rm TE}(\la^*_{i},0.5 )$), $i = 1, 2,3 $.
Further, suppose that $I_{p_1}, I_{p_2}, I_{p_3}$  ($I_{p^*_1},
I_{p^*_2}, I_{p^*_3} $) is a set of independent Bernoulli random
variables, independent of the $X_{\la_i}$'s ($X_{\la^*_i}$'s),
with ${\rm E}[I_{p_i}]=p_i$ (${\rm E}[I_{p^*_i}]=p^*_i$),
$i=1,2,3$. Also, suppose that
$(\la_1,\la_2,\la_3)=(-0.7,0.8,-0.9)$,
$(\la^*_1,\la^*_2,\la^*_3)=(-0.1806,0.0896,-0.7090)$,
$(p_1,p_2,p_3)=(0.4,0.2,0.7)$, and
$(p^*_1,p^*_2,p^*_3)=(0.4345,0.3698,0.4711)$. Take
$h(p)=\log(2+p)$ and the $T$-transform matrices with the different
structures as

$T_{0.9}=\begin{bmatrix}
1&0 & 0 \\
0&0.9 & 0.1\\
0&0.1 & 0.9\\
\end{bmatrix}$, $T_{0.3}=\begin{bmatrix}
0.3&0 & 0.7 \\
0&1 & 0\\
0.7&0 & 0.3\\
\end{bmatrix}$ and $T_{0.6}=\begin{bmatrix}
0.6&0.4 & 0 \\
0.4&0.6 & 0\\
0&0 & 1\\
\end{bmatrix}$.

It can be easily verified that $(\lab,h(\pb))$,
$(\lab,h(\pb))T_{0.9}$ and $(\lab,h(\pb))T_{0.9} T_{0.3}$ are in
$S_3$, and $(\lab^*,h(\pb^*))=(\lab,h(\pb))T_{0.9} T_{0.3}
T_{0.6}$. Thus, Corollary \ref{cor1} implies $Y^*_{3:3}\leq_{\rm
st}Y_{3:3}$.

\item Let $ X_{\la_1}, X_{\la_2} , X_{\la_3} $ ($ X_{\la^{*}_1},
X_{\la^{*}_2} , X_{\la^{*}_3} $) be independent random variables
with $ X_{\la_i} \thicksim {\rm TE}(\la_i,2)$ ($ X_{\la^{*}_i}
\thicksim {\rm TE}(\la^*_{i},2 )$), $i = 1, 2,3 $. Further,
suppose that $I_{p_1}, I_{p_2}, I_{p_3}$  ($I_{p^*_1}, I_{p^*_2},
I_{p^*_3} $) is a set of independent Bernoulli random variables,
independent of the $X_{\la_i}$'s ($X_{\la^*_i}$'s), with ${\rm
E}[I_{p_i}]=p_i$ (${\rm E}[I_{p^*_i}]=p^*_i$), $i=1,2,3$. Also,
suppose that $(\la_1,\la_2,\la_3)=(0.1,0.3,-0.6)$,
$(\la^*_1,\la^*_2,\la^*_3)=(0.5,-0.3,0.1)$,
$(p_1,p_2,p_3)=(0.5,0.3,0.7)$, and
$(p^*_1,p^*_2,p^*_3)=(0.3,0.9,0.1)$. It can be easily verified
that the conditions of Theorem \ref{th5} hold and so we can
conclude that $Y^*_{1:3}\leq_{\rm st}Y_{1:3}$.

Figure 1 (top panels) represents the survival functions
of $Y_{1:3}$, $Y^*_{1:3}$, $Y_{3:3}$ and $Y^*_{3:3}$ 
for the transmuted exponential distribution.
\end{itemize}

\subsection{Transmuted-G Weibull distribution}
Suppose that the baseline distribution in transmuted-G model is
Weibull distribution with shape parameter $\alpha$ and scale
parameter $\beta$. Here this distribution is denoted by ${\rm TW}
(\la,\alpha,\beta)$. For more details on this distribution, we
refer to Aryal and Tsokos  \cite{arts} and Khan et al.
\cite{khan}.

\begin{itemize}
\item  Let $ X_{\la_1}, X_{\la_2} , X_{\la_3} $ ($ X_{\la^{*}_1},
X_{\la^{*}_2} , X_{\la^{*}_3} $) be independent random variables
with $ X_{\la_i} \thicksim {\rm TW}(\la_i,0.3,1.5)$ ($
X_{\la^{*}_i} \thicksim {\rm TW}(\la^*_{i},0.3,1.5 )$), $i = 1,
2,3 $. Further, suppose that $I_{p_1}, I_{p_2}, I_{p_3}$
($I_{p^*_1}, I_{p^*_2}, I_{p^*_3} $) is a set of independent
Bernoulli random variables, independent of the $X_{\la_i}$'s
($X_{\la^*_i}$'s), with ${\rm E}[I_{p_i}]=p_i$ (${\rm
E}[I_{p^*_i}]=p^*_i$), $i=1,2,3$. Also, suppose that
$(\la_1,\la_2,\la_3)=(0.7,0.3,-0.9)$,
$(\la^*_1,\la^*_2,\la^*_3)=(0.1544,-0.5464,0.4920)$,
$(p_1,p_2,p_3)=(0.1,0.4,0.8)$, and
$(p^*_1,p^*_2,p^*_3)=(0.3506,0.6295,0.2124)$. Take
$h(p)=\frac{5p+2}{p+1}$ and the $T$-transform matrices with the
different structures as

$T_{0.1}=\begin{bmatrix}
1&0 & 0 \\
0&0.1 & 0.9\\
0&0.9 & 0.1\\
\end{bmatrix}$, $T_{0.4}=\begin{bmatrix}
0.4&0 & 0.6 \\
0&1 & 0\\
0.6&0 & 0.4\\
\end{bmatrix}$ and $T_{0.8}=\begin{bmatrix}
0.8&0.2 & 0 \\
0.2&0.8 & 0\\
0&0 & 1\\
\end{bmatrix}$.

It can be easily verified that $(\lab,h(\pb))$,
$(\lab,h(\pb))T_{0.1}$ and $(\lab,h(\pb))T_{0.1} T_{0.4}$ are in
$S_3$, and $(\lab^*,h(\pb^*))=(\lab,h(\pb))T_{0.1} T_{0.4}
T_{0.8}$. Thus, Corollary \ref{cor1} implies $Y^*_{3:3}\leq_{\rm
st}Y_{3:3}$.

\item Let $ X_{\la_1}, X_{\la_2} , X_{\la_3} $ ($ X_{\la^{*}_1},
X_{\la^{*}_2} , X_{\la^{*}_3} $) be independent random variables
with $ X_{\la_i} \thicksim {\rm TW}(\la_i,2,0.6)$ ($ X_{\la^{*}_i}
\thicksim {\rm TW}(\la^*_{i},2,0.6 )$), $i = 1, 2,3 $. Further,
suppose that $I_{p_1}, I_{p_2}, I_{p_3}$  ($I_{p^*_1}, I_{p^*_2},
I_{p^*_3} $) is a set of independent Bernoulli random variables,
independent of the $X_{\la_i}$'s ($X_{\la^*_i}$'s), with ${\rm
E}[I_{p_i}]=p_i$ (${\rm E}[I_{p^*_i}]=p^*_i$), $i=1,2,3$. Also,
suppose $(\la_1,\la_2,\la_3)=(0.3,0.7,0.5)$,
$(\la^*_1,\la^*_2,\la^*_3)=(0.8,0.4,0.5)$,
$(p_1,p_2,p_3)=(0.6,0.3,0.2)$, and
$(p^*_1,p^*_2,p^*_3)=(0.4,0.5,0.1)$. It can be easily verified
that the conditions of Theorem \ref{th5} hold. Thus we can
conclude that $Y^*_{1:3}\leq_{\rm st}Y_{1:3}$.

Figure 1 (bottom panels) represents survival
functions of $Y_{1:3}$, $Y^*_{1:3}$, $Y_{3:3}$ and $Y^*_{3:3}$ for the transmuted Weibull distribution.
\end{itemize}

\begin{figure}[tbp]\label{fig1}
\centerline{ \epsfig{file=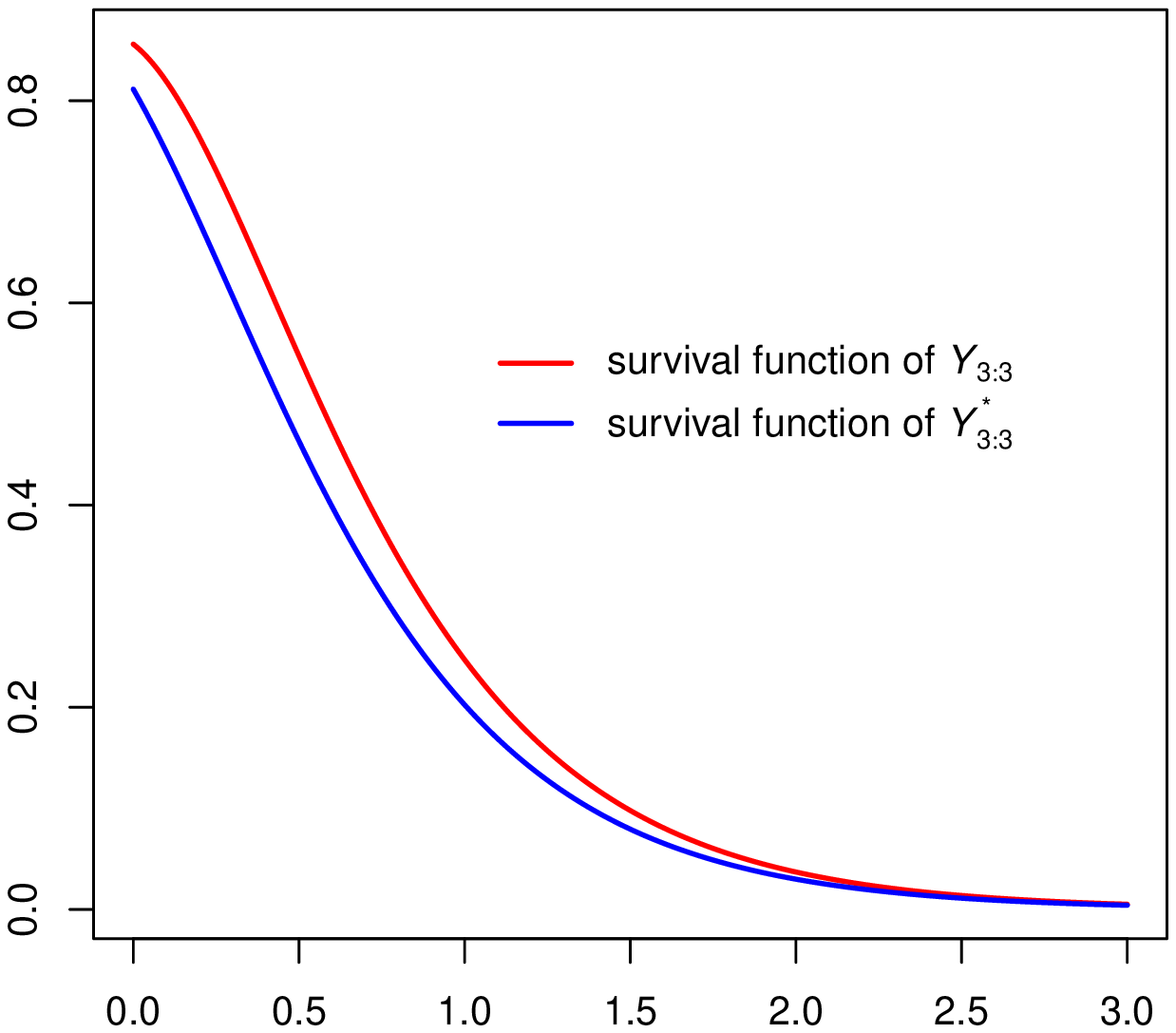,height=2.5in,width=2.5in}
\epsfig{file=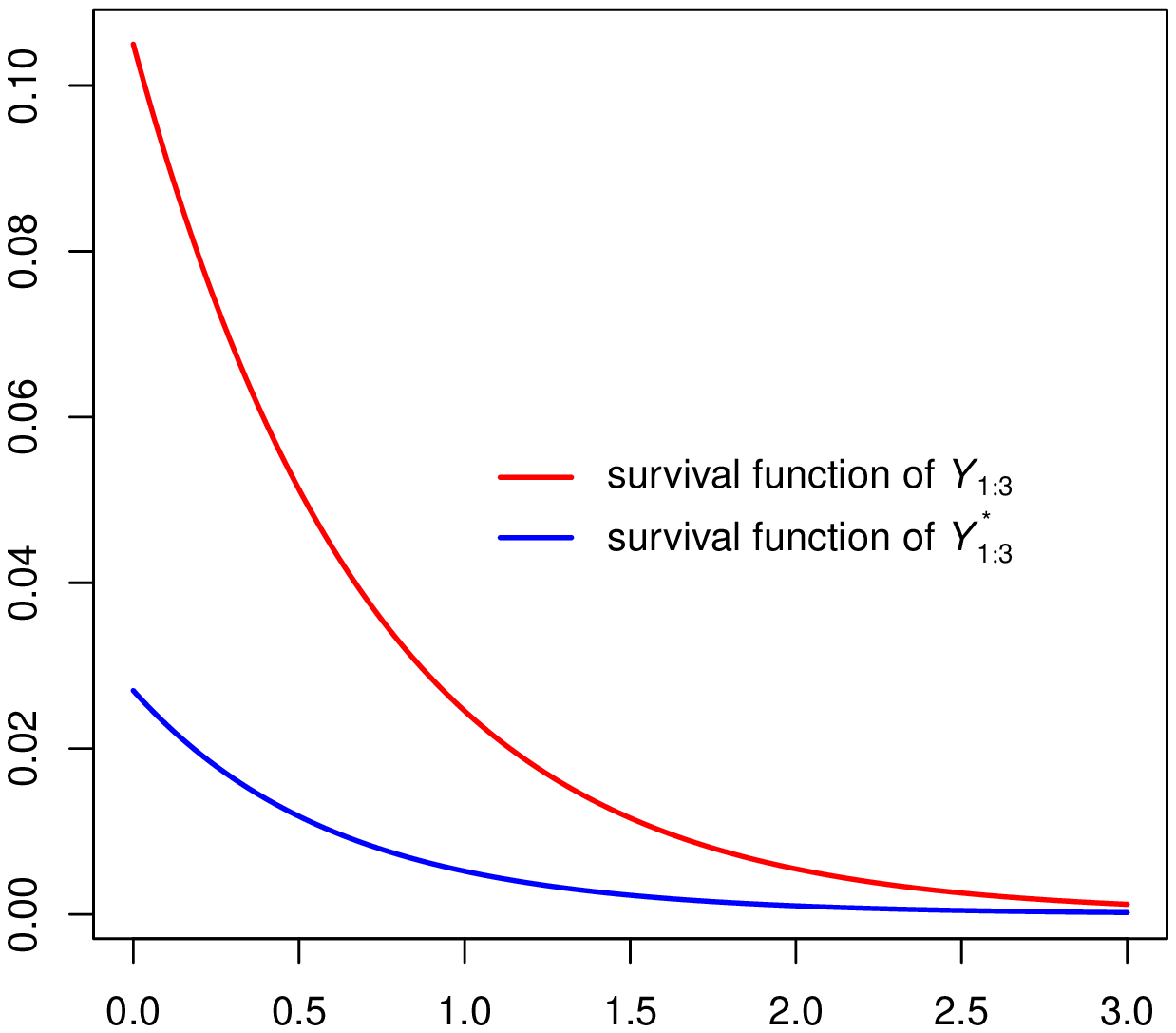,height=2.5in, width=2.5in}}
\centerline{\epsfig{file=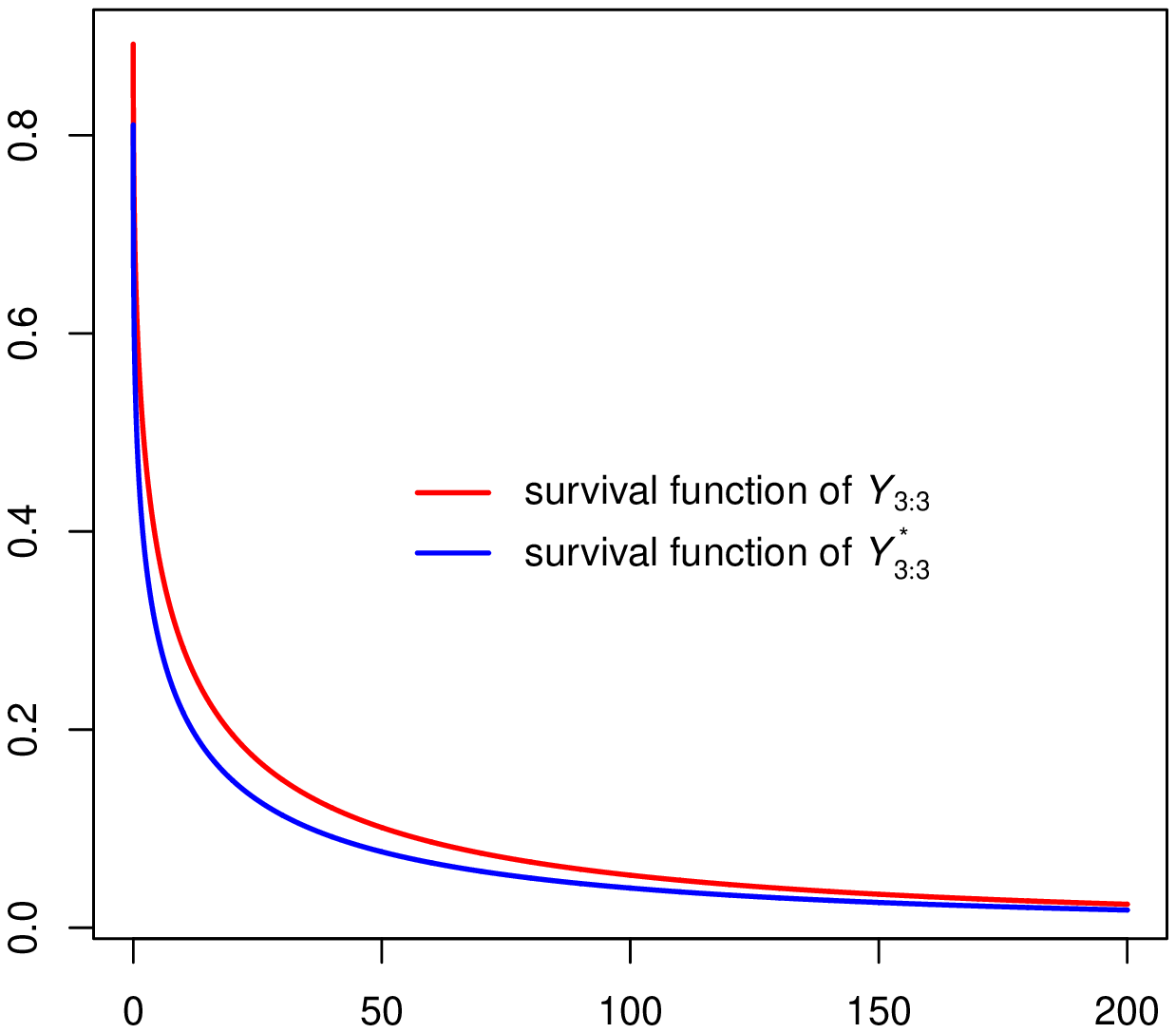,height=2.5in, width=2.5in}
\epsfig{file=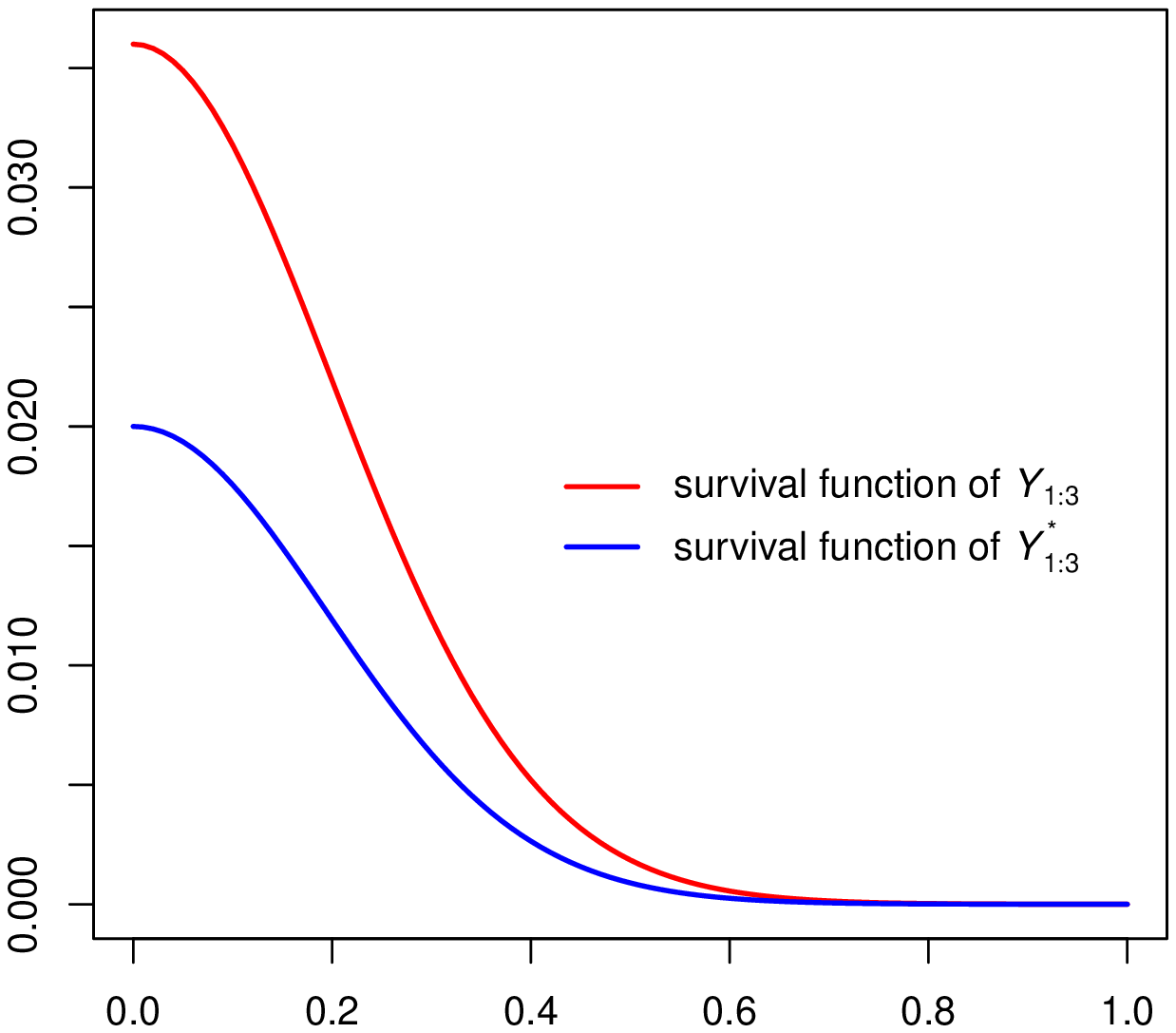,height=2.5in, width=2.5in}}
\par
\par
\vspace{-0.4cm} \caption{\footnotesize{Plots of the survival
functions of $Y_{1:3}$, $Y^*_{1:3}$, $Y_{3:3}$ and $Y^*_{3:3}$ in transmuted exponential distribution (top
panels) and transmuted Weibull distribution (bottom panels).}}
\end{figure}

\section*{Conclusion}
In this paper, under some certain conditions, we discussed stochastic comparisons between the largest claim amounts in the sense of usual
stochastic ordering and reversed hazard rate ordering and stochastic comparisons between the smallest claim amounts in the sense of usual
stochastic ordering, hazard rate ordering and dispersive ordering in transmuted-G model. However, we applied some established results for two special cases of transmuted-G model, such as the transmuted exponential distribution and the transmuted Weibull distribution.
It is very important to mention that the conditions of the most established results do not depend on the baseline distribution properties.




\end{document}